\documentclass{article}
\usepackage{geometry}
\usepackage{enumerate}
\usepackage{amssymb,amsfonts,amsmath,bbm,mathrsfs,stmaryrd,mathtools,mathabx}
\usepackage{amsmath}
\usepackage{amsfonts,amscd, amssymb}
\usepackage{xcolor}
\usepackage{tikz-cd}
\usepackage{amsthm,latexsym,euscript,amscd, braket}
\usepackage{aliascnt}
\usepackage{amsmath,amsthm}
\usepackage{thmtools, ulem}
\usepackage{xcolor,colortbl}
\definecolor{lavender}{rgb}{0.9, 0.9, 0.98}
\definecolor{UGARed}{rgb}{0.729,0.047,0.184}
\usepackage{url}
\usepackage{authblk}
\usepackage{enumerate}
\usepackage{thmtools}
\usepackage[backref=page]{hyperref}
\newcommand\nnfootnote[1]{%
  \begin{NoHyper}
  \renewcommand\thefootnote{}\footnote{#1}%
  \addtocounter{footnote}{-1}%
  \end{NoHyper}
}
\usepackage{cleveref}
\creflabelformat{enumi}{#2(#1)#3}
\crefformat{enumi}{(#2#1#3)}
\crefformat{equation}{(#2#1#3)} 
\Crefformat{equation}{(#2#1#3)}

\crefformat{inequality}{(#2#1#3)} 
\Crefformat{Inequality}{(#2#1#3)}
\hypersetup{colorlinks,
linkcolor=UGARed,
citecolor=blue,
pdftitle={},
bookmarksopen=true,
bookmarksnumbered=true}         
\usepackage{tikz}
\usepackage{graphicx} 
 \usepackage{todonotes} 
\theoremstyle{plain} 
\newtheorem{thm}{Theorem}[section]
\newtheorem{cor}[thm]{Corollary}
\newtheorem{prop}[thm]{Proposition}
\newtheorem{defn}[thm]{Definition}

\newtheorem{rmk}[thm]{Remark}

\numberwithin{equation}{section}

\newcommand{\zw}[1]{{\color{UGARed}{#1\textsuperscript{ZW}}}}

\DeclareMathOperator{\Cdb}{{\mathbb C}}

\DeclareMathOperator{\Hdb}{{\mathbb H}}

\DeclareMathOperator{\Rdb}{{\mathbb R}}

\DeclareMathOperator{\Al}{{\mathcal A}}

\newcommand{\norm}[1]{\left\Vert#1\right\Vert}

\begin{document}
\title{Error Estimates and Higher Order Trotter Product Formulas in Jordan-Banach Algebras}


\author[1, 2]{Sarah Chehade \thanks{\texttt{sarah-chehade@utc.edu}}}
\author[3]{Andrea Delgado \thanks{\texttt{delgado.andrea.21@gmail.com}}}
\author[4]{Shuzhou Wang \thanks{\texttt{szwang@uga.edu}}}
\author[4, 5]{Zhenhua Wang \thanks{\texttt{ezrawang@uga.edu}}}

\affil[1]{Computational Sciences and Engineering Division, Oak Ridge National Laboratory, Oak Ridge, TN 37831}
\affil[2]{UTC Quantum Center, University of Tennessee, Chattanooga, TN 37403}
\affil[3]{Physics Division, Oak Ridge National Laboratory, Oak Ridge, TN 37831}
\affil[4]{Department of Mathematics, University of Georgia, Athens, GA 30602}
\affil[5]{Department of Mathematics, Alabama A$\&$M University, Huntsville, AL 35811}

\date{\today}
\maketitle

\begin{abstract}
     In quantum computing, Trotter estimates are critical for enabling efficient simulation of quantum systems and quantum dynamics, help implement complex quantum algorithms, and provide a systematic way to control approximate errors. 
In this paper, we extend the analysis of Trotter-Suzuki approximations, including third and higher orders, to Jordan-Banach algebras. We solve an open problem in our earlier paper on the existence of second-order Trotter formula error estimation in Jordan-Banach algebras. To illustrate our work, we apply our formula to simulate Trotter-factorized spins, and show improvements in the approximations. Our approach demonstrates the adaptability of Trotter product formulas and estimates to non-associative settings, which offers new insights into the applications of Jordan algebra theory to operator dynamics. 

\end{abstract}

\nnfootnote{\textit{Mathematics Subject Classification \rm{(2020)}:} {Primary
	17C90,
        81P45,
        15A16;
	Secondary
	17C65, 
	81R15,
        46H70
	}}

\nnfootnote{\textit{Key words:} Trotter Product formula, Error estimate, Higher order approximation, Jordan-Banach algebra}
\vspace{0.25cm}
 \nnfootnote{This manuscript has been partially authored by UT-Battelle, LLC under Contract No. DE-AC05-00OR22725 with the U.S. Department of Energy. The United States Government retains and the publisher, by accepting the article for publication, acknowledges that the United States Government retains a non-exclusive, paid-up, irrevocable, world-wide license to publish or reproduce the published form of this manuscript, or allow others to do so, for United States Government purposes. The Department of Energy will provide public access to these results of federally sponsored research in accordance with the DOE Public Access Plan (http://energy.gov/downloads/doe-public-access-plan).}
\section{Introduction}
Jordan algebras, a class of non-associative algebras, introduced by Pascual Jordan in the 1930s, have significant applications in various areas of mathematics, physics, and computer science. In particular, in both classical and quantum mechanics, observables naturally form a Jordan algebra. Originally formulated to provide formal and algebraic methods to describe quantum mechanics \cite{Jordan:Pascual:vonNeumann}, Jordan algebras have since emerged as a potentially useful mathematical methodology in quantum information science \cite{wang:wang:2021:means, wang:wang:2022:refined, wang:wang:2022:entropy}.

Trotter product formula, developed by Hale Trotter \cite{trotter:1959:product}, is a fundamental result in the mathematical theory of quantum dynamics. It provides a rigorous framework for approximating the exponential of a sum of non-commuting operators. 
This formula underpins the study of time-evolution in quantum systems and offers a pathway to discretizing continuous dynamics, a critical tool in numerical simulations of quantum mechanics.
Building on this work, Suzuki developed advanced error bounds and fractal decompositions for operator exponentials, enabling higher-order approximations and reducing error rates in simulations \cite{ suzuki:1976:generalized, suzuki:1985:transfer, suzuki:1990:fractal, suzuki:1991:general}. Suzuki's methods, known as Suzuki estimates, have found broad applications in quantum chemistry, condensed matter physics, and quantum computing \cite{Poulin:etal:chemistry}.
More recently, Childs and collaborators expanded the theory by introducing refined error bounds based on commutator scaling, providing deeper insights into the structure of Trotter error \cite{childs:etal:trotter}. This work bridges theoretical and practical aspects of quantum simulation, offering guidelines for optimizing algorithms on near-term and fault-tolerant quantum devices. 

The extension of Trotter product formulas to the framework of Jordan-Banach algebras was first carried out by Peralta and collaborators to study spectral-valued multiplicative functionals \cite{ escolano:peralta:villena:2024:lie}. 
Motivated by their work, we established Trotter product formula error estimations in JB-algebras \cite{chehade:wang:wang:2024:suzuki}, which is applicable to quantum dynamics, and extended Trotter product formulas for three elements to any finite number of elements in Jordan-Banach algebras. This expansion established a quadratic speedup in Trotter error estimates in JB-algebras. This theory has found applications in quantum information science \cite{wang:chehade:dumitrescu:2024:semi}. There, the advanced error estimates and algebraic insights underpin models of quantum processes with symmetry constraints, enabling more precise descriptions of dynamics and interactions. These applications illustrate the impact of merging Jordan algebra methods with Trotter approximations, bridging abstract mathematics and real-world quantum technologies.

In this manuscript, 
we develop higher-order Jordan-Trotter product formulas, 
and establish second and third order Trotter error estimates in Jordan-Banach algebras. 
As an illustration, we apply our formula to simulate Trotter-factorized spins, and show improvements in the approximations. The structure of the paper is as follows: \Cref{section:preliminaries} recalls the mathematical preliminaries and sets up notation. In \Cref{sec:2order}, we establish second-order Trotter error estimates in Jordan-Banach algebras, building on and extending the authors' earlier work on JB-algebras in \cite{chehade:wang:wang:2024:suzuki}. This section also solves the open problem on the existence of second-order Trotter error estimates in Jordan-Banach algebras, as mentioned in the introduction of that paper. In \Cref{sec:highOrder}, we first develop a third order Jordan-Trotter product formula and establish its error estimate in Jordan-Banach algebras. Then, we construct higher-order Trotter product formulas in Jordan-Banach algebras. Finally, \Cref{sec:example} highlights the implications of the third order Jordan-Trotter product formula in quantum simulation, demonstrating how our methodology enables more accurate and efficient modeling of quantum systems.

\section{Preliminaries}\label{section:preliminaries}
For convenience of the reader, in this section, we give some background on Jordan-Banach algebras, JB*-algebras, and fix the notation. For more information, we refer the reader to \cite{alfsen:2003:geometry,cabrera:rodriguez:2014:non, hanche:1984:jordan}.

\begin{defn}\label{defn:ja}
	A {\bf Jordan algebra} $\Al$ over real or complex numbers is a vector space $\Al$ over $\Rdb$ or $\Cdb$ equipped with a  bilinear product $\circ$ that satisfies the following identities:
	\[A\circ B =B\circ A, \,\ \,\ (A^2\circ B)\circ A=A^2\circ (B\circ A),\]
where $\displaystyle A^2$ means $\displaystyle A\circ A.$ 
\end{defn}
Any associative algebra $\Al$ admits an underlying Jordan algebra structure with Jordan product given by 
\begin{align}\label{notion:specialJ}
 A\circ B=(AB+BA)/2.   
\end{align}	
	Jordan subalgebras of such underlying Jordan algebras are called {\bf special}. 
\begin{defn}\label{defn:triple}
	The Jordan triple product of any three elements $A, B, C$ in  a Jordan algebra $\Al$ is defined by 
	\begin{align}\label{notion:triple}
	\{A, B, C\}&:=(A\circ B)\circ C+(B\circ C)\circ A-(C\circ A)\circ B.
			\end{align}
\end{defn}
Note that in some literature, for instance \cite{alfsen:2003:geometry, hanche:1984:jordan}, $\displaystyle \{A, B, C\}$ is denoted by $\{ABC\}.$

\begin{defn}\label{defn:uop} 
With the notation in {\rm \Cref{defn:triple}}, we define an operator $U_{A, C}:\Al \to \Al$ by 
\begin{align}\label{notion:uop}
U_{A, C}(B):=\{A,B, C\}.
\end{align}
Denote $U_{A, A}$ by $U_A.$
If $\Al$ is special, then $U_A(B)=ABA.$ 

We also define the Jordan multiplication operator $M_A:\Al \to \Al$ by
\begin{align}\label{notion:mtop}
M_A(B):=A\circ B \, (=B\circ A).
\end{align}
\end{defn}

\begin{defn}
    A \textbf{Jordan-Banach algebra} $\mathcal{A}$ is a Jordan algebra with a complete norm satisfying $\Vert A\circ B\Vert \leq \|A\|\|B\|$, for all $A,B\in\mathcal{A}$. 
\end{defn}

The set of bounded self adjoint operators on a Hilbert space $\Hdb$, denoted by $B(\Hdb)_{sa}$ is a Jordan Banach-algebra, which is not associative. It is an important object in physics, as it is the set of observables in a quantum mechanical system. 

We have the following straightforward norm estimates for $U_{A, C}$, $U_A$ and $M_A$ in a Jordan-Banach algebra.
\begin{prop}\label{prop:uopmop}
For $A, C$ in a Jordan-Banach algebra $\Al,$ 
\begin{align}
   \Vert U_{A, C}\Vert \le 3\Vert A\Vert\Vert C\Vert, \quad
   \Vert U_{A}\Vert \le 3\Vert A\Vert^2, \quad  \Vert M_A\Vert \le \Vert A\Vert.
\end{align}
\end{prop}

\begin{defn}\label{defn:JB*}   
    A \textbf{JB*-algebra} $(\Al,*)$ is a complex Jordan-Banach algebra equipped with an involution $*$ satisfying that for all $A,B\in\Al$, the following condition holds 
    \begin{align*}
     \Vert A\circ B\Vert\leq \Vert A\Vert \Vert B\Vert,\quad  \quad \Vert \{A, A^*, A\}\Vert=\Vert A\Vert^3.    
    \end{align*}    
    An element $U$ in $\mathcal{A}$ is said to be \textbf{unitary} if 
    \begin{align*}
    U\circ U^*=I, \quad U^2\circ U^*=U.  
    \end{align*}
\end{defn}
Note that if $A$ is self-adjoint, $\displaystyle \exp(iA)$ is a unitary. 


\section{Second-Order Jordan-Trotter Product Formulas}\label{sec:2order} 

   For $\displaystyle A_1, A_2,\cdots, A_m$ in a unital Jordan-Banach algebra $\mathcal{A}$ and $z \in\Cdb$, define
    \begin{align} \label{notion:G}
    G(z)&:=\exp\left(z\sum_{k=1}^m A_k\right)
    \\ 
    \label{notion:J2}
    J_2(z)&:=M_{\exp(z A_m)}M_{\exp(z A_{m-1})}\cdots M_{\exp(z A_2)}(\exp(z A_1)) \\
    \label{notion:S2}
    S_2(z)&:=U_{\exp(\frac{z}{2}A_m)}U_{\exp(\frac{z}{2}A_{m-1})}\cdots U_{\exp(\frac{z}{2}A_2)}(\exp(zA_1))\\
    \label{notion:T2}
    T_2(z)&:=I+z\left( \sum_{k=1}^m A_k\right)+\frac{z^2}{2!}\left( \sum\limits_{k=1}^m A_k \right)^2. 
    \end{align}

Furthermore, if $m$ is odd, i.e. $m=2p+1$ for some positive integer $p$, then we define
\begin{align}\label{notion:QS2}
Q_{S_{2}}(z)&:=U_{\exp(zA_{2p}), \exp(zA_{2p+1})}U_{\exp(zA_{2p-2}), \exp(zA_{2p-1})}\cdots U_{\exp(zA_2), \exp(zA_3)}(exp(zA_1)).
\end{align}

Building upon our results in \cite{chehade:wang:wang:2024:suzuki}, \Cref{thm:Trotter:first:JordanBanach}, \Cref{thm:Trotter:quasisymmetric:JordanBanach} and \Cref{coro:Trotter:second:JordanBanach} below provide an affirmative answer to the question therein whether 
Trotter-Suzuki estimates exist in the Jordan-Banach algebra setting.

\begin{thm} \label{thm:Trotter:first:JordanBanach}
For any finite number of elements $A_1, A_2,\cdots, A_m$ in a unital Jordan-Banach algebra $\mathcal{A},$ we have the following estimate 
\begin{align}\label{ineq:Trotter:first:JordanBanach}
    		\left\Vert \exp \left( \sum_{k=1}^m A_j\right)-\left[J_2\left(\frac{1}{n}\right)\right]^n\right \Vert	
    		\leq
    		\dfrac{1}{3n^2}\left( \sum_{k=1}^m \Vert A_j\Vert\right)^3\exp\left(\sum_{k=1}^m {\Vert A_j\Vert}\right).	
    \end{align} 	
\end{thm}
\begin{proof}
Let 
\[
E = G\left(\frac{1}{n} \right), \quad F = J_2\left(\frac{1}{n}\right),\quad \mbox{and} \quad T_2=T_2\left(\frac{1}{n}\right) .
\]
Then
\[
\|E\| \leq \exp\left(\frac{1}{n} \sum_{k=1}^m \|A_k\|\right) \quad \text{and} \quad \|F\| \leq \exp\left(\frac{1}{n} \sum_{k=1}^m \|A_k\|\right).
\]

Note that $\displaystyle \exp \left( \sum_{k=1}^m A_j\right)-\left[J_2\left(\frac{1}{n}\right)\right]^n=E^n-F^n$ and
\begin{align}
 E^n-F^n &=\sum_{j=0}^{n-1}M_F^j(E^{n-j})-\sum_{k=1}^nM_F^k(E^{n-k})=\sum_{j=0}^{n-1}M_F^j(E^{n-j})-M_F^{j+1}(E^{n-1-j}) \nonumber\\
 &=\sum_{j=0}^{n-1}M_F^j(E^{n-j}-F\circ E^{n-1-j})=\sum_{j=0}^{n-1}M_F^j\left((E-F)\circ E^{n-1-j}\right)\nonumber
\end{align}

Thus, 
\begin{align}
\Vert E^n-F^n\Vert&=\Vert \sum_{j=0}^{n-1}M_F^j\left((E-F)\circ E^{n-1-j}\right) \Vert  \leq \Vert E-F\Vert\left(\sum_{j=0}^{n-1}\Vert F\Vert^j \Vert E \Vert^{n-1-j}\right) \nonumber\\
 &\leq n \Vert E-F\Vert\left(\max\{\Vert E\Vert, \Vert F\Vert \}\right)^{n-1}\leq n \Vert E-F\Vert\exp\left(\frac{n-1}{n}\sum_{k=1}^m \Vert A_k\Vert\right) \nonumber\\
& \leq n \left(\left\Vert E-T_2\right\Vert+\left\Vert F-T_2\right \Vert\right )\exp\left(\frac{n-1}{n}\sum_{k=1}^m \Vert A_k\Vert\right). \label{ineq:En-Fn:error}
\end{align}

By induction, the second-order Taylor polynomial for $E$ is equal to that for $F$, both of which are $T_2$. Applying an argument similar to that used in the proof of Theorem 3.1 in \cite{chehade:wang:wang:2024:suzuki}, we have
\begin{align}
   \norm{E-T_2}&\leq \frac{1}{3!\cdot n^3}\left(\sum_{k=1}^m \Vert A_k\Vert \right)^3\exp \left(\frac{1}{n}\sum_{k=1}^m \Vert A_k\Vert \right) \label{ineq:E-T2:error}  \\
   \norm{F-T_2}&\leq \frac{1}{3!\cdot n^3}\left(\sum_{k=1}^m \Vert A_k\Vert \right)^3\exp \left(\frac{1}{n}\sum_{k=1}^m \Vert A_k\Vert \right). \label{ineq:F-T2:error}
\end{align}
The desired result follows by combining \cref{ineq:En-Fn:error} \cref{ineq:E-T2:error} and \cref{ineq:F-T2:error}. 
\end{proof}

Note that the technique used in inequality (3.2) of \cite{chehade:wang:wang:2024:suzuki} can not be applied to obtain inequality \Cref{ineq:En-Fn:error} here because we do not have an embedding of the Jordan-Banach algebra into a JC-algebra or JC*-algebra. 
In fact, \Cref{thm:Trotter:first:JordanBanach} presents an alternative approach to proving the Lie-Trotter formulas in Jordan-Banach algebras, utilizing methods distinct from those used by Escolano et al. in \cite{escolano:peralta:villena:2024:lie}.

\begin{prop}\label{prop:Trotter:first:JordanBanach}
Let $A_1, A_2,\cdots, A_m$ be any finite number of elements in a unital Jordan-Banach algebra $\mathcal{A}$ and $t\in \Rdb$. Then
\begin{align*}
 \exp\left(t\sum_{k=1}^m A_k\right)=J_2(t)+O(t^3).   
\end{align*}
\end{prop}
\begin{proof}
The result is derived from the fact that the 2nd-order Taylor polynomial for $\displaystyle \exp\left(t\sum_{k=1}^m A_k\right)$ coincides with that of $J_2(t)$, which is $T_2(t)$.
\end{proof}

The following result extends \cite[Theorem 2.3]{wang:2023:suzuki} to the context of Jordan-Banach algebras.
\begin{thm}\label{thm:Trotter:quasisymmetric:JordanBanach}
 Let $A_1, A_2,\cdots, A_{2p+1}$ be any finite odd number of elements in a unital Jordan-Banach algebra $\mathcal{A}$ and let $t\in \Rdb$. Then
 \[
\left\Vert\exp\left(t\sum_{k=1}^{2p+1} A_k\right)-Q_{S_2}(t)\right \Vert \leq  \frac{3^{p}+1}{6}\vert t\vert^3\left(\sum_{k=1}^{2p+1}\Vert A_k\Vert\right)^3\exp\left(\vert t\vert\sum_{k=1}^{2p+1} \Vert A_k\Vert\right)
\]
and 
\[
\exp\left(t\sum_{k=1}^{2p+1} A_k\right)=Q_{S_2}(t)+O(t^3).	
\]
\end{thm}
\begin{proof}
For any $1\leq j\leq p,$ define
\begin{align}\label{notion:Qj}
W_{j}(t)&=U_{\exp(tA_{2j}), \exp(tA_{2j+1})}\cdots U_{\exp(tA_2), \exp(tA_3)}(exp(tA_1)) \\
\label{notion:Vj}
 \widetilde{T_j}(t)&=I+t\left(A_1+A_2+\cdots+A_{2j+1}\right)+\frac{t^2}{2!}\left(A_1+A_2+\cdots+A_{2j+1}\right)^2.
\end{align}
Note that 
\begin{align*}
\left\Vert\exp\left(t\sum_{k=1}^{2p+1} A_k\right)-Q_{S_2}(t)\right \Vert\leq 	\left\Vert\exp\left(t\sum_{k=1}^{2p+1} A_k\right)- \widetilde{T_p}(t)\right\Vert +\left\Vert Q_{S_2}(t)-\widetilde{T_p}(t)\right\Vert. 
\end{align*}

By a similar argument as in the proof of \cite[Theorem 3.3(i)]{chehade:wang:wang:2024:suzuki}, with $\displaystyle \frac{1}{n}$ replaced by $t$, we obtain
\begin{align*}
\left\Vert\exp\left(t\sum_{k=1}^{2p+1} A_k\right)- \widetilde{T_p}(t)\right\Vert\leq \frac{1}{6}\vert t\vert^3\left(\sum_{k=1}^{2p+1} \Vert A_k\Vert \right)^3\exp\left(\vert t\vert\sum_{k=1}^{2p+1} \Vert A_k\Vert\right).	
\end{align*}

Since 
\begin{align}
W_1(t)=[\exp(tA_2)\circ\exp(tA_1)]\circ \exp(tA_3)+ [\exp(tA_1)\circ \exp(tA_3)] \circ \exp(tA_2) \nonumber
\\ -[\exp(tA_3)\circ\exp(tA_2 )]\circ \exp(tA_1), \label{eqt:notion:W1}	
\end{align}
it follows that  \Cref{prop:Trotter:first:JordanBanach} implies that  the second order Taylor polynomial of $W_1(t)$ is $ \widetilde{T_1}(t).$ By similar computation, the second order Taylor polynomial for $W_j(t)$ is $ \widetilde{T_j}(t)$ for any $1\leq j\leq p$.

According to \Cref{eqt:notion:W1}, for any integer $s>2,$ the norm of the sum of all terms of degree $s$ of Taylor expansion of $W_1(t)\, \leq$ the sum of all terms of degree $s$ of Taylor expansion of $\displaystyle 3\cdot \exp\left(\vert t\vert\sum_{k=1}^3\Vert A_k\Vert\right ).$ 
  $W_1(t).$ Let $P_s(t)$ be the sum of all terms of degree $s$ of Taylor expansion of $W_1(t)$. Then
	\begin{align*}
		\left\Vert W_1(t)-\widetilde{T_1}(t)\right\Vert &=\left\Vert \sum_{s>2} P_s(t)\right \Vert \leq 3\left[\exp \left(\vert t\vert \sum_{k=1}^3 \Vert A_k\Vert \right)-\left(I+\vert t\vert \sum_{k=1}^3 \Vert A_k\Vert+\frac{t^2}{2!}\left(\sum_{k=1}^3 \Vert A_k\Vert\right)^2 \right) \right]
	\end{align*}

More generally, we obtain 
\begin{align*}
\Vert Q_{S_2}(t)-\widetilde{T_p}(t)\Vert&=\Vert W_p(t)-\widetilde{T_p}(t)\Vert\\
&\leq 3^p\left[\exp \left(\vert t\vert \sum_{k=1}^{2p+1} \Vert A_k\Vert \right)-\left(I+\vert t\vert \sum_{k=1}^{2p+1} \Vert A_k\Vert+\frac{t^2}{2!}\left(\sum_{k=1}^{2p+1} \Vert A_k\Vert\right)^2 \right) \right]\\
&\leq \frac{3^p}{3!}\vert t\vert^3\left(\sum_{k=1}^{2p+1} \Vert A_k\Vert\right)^3\exp\left(\sum_{k=1}^{2p+1} \Vert A_k\Vert\right).
\end{align*}
Thus, the desired result follows directly.
\end{proof}

\begin{cor}\label{coro:Trotter:second:JordanBanach}
Let $A_1, A_2,\cdots, A_m$ be any finite number of elements in a unital Jordan-Banach algebra $\mathcal{A}$ and $t\in \Rdb$.
Then 
\begin{align*}
\left\Vert\exp\left(t\sum_{k=1}^m A_k\right)-S_2(t)\right \Vert \leq  \frac{3^{m-1}+1}{6}\vert t\vert^3\left(\sum_{k=1}^m\Vert A_k\Vert\right)^3\exp\left(\vert t\vert\sum_{k=1}^m \Vert A_k\Vert\right)
\end{align*}
and
\begin{align*}
 \exp\left(t\sum_{k=1}^m A_k\right)=S_2(t)+O(t^3).   
\end{align*}
\end{cor}

\begin{proof}
 The result is obtained by applying \Cref{thm:Trotter:quasisymmetric:JordanBanach} to the sequence $\displaystyle \{B_1, B_2,\dots, B_{2m-1}\}$ where $B_1=A_1, B_2=B_3=A_2/2,\cdots, B_{2m-2}=B_{2m-1}=A_m/2.$
\end{proof}

In the context of JB*-algebras, we establish sharp error bounds for Trotter product formulas involving unitary elements.
\begin{thm}\label{thm:Trotter:first:JB*}
    For any finite number of self-adjoint elements $A_1, A_2,\cdots, A_m$ in a unital JB*-algebra $\mathcal{A},$ 
    \begin{align}\label{ineq:Trotter:first:JB*}
    		\left\Vert \exp \left( i\bigg(\sum_{k=1}^m A_k\bigg)\right)-\left[J_2\left(\frac{i}{n}\right)\right]^n\right \Vert	
    		\leq
    		\dfrac{1}{3n^2}\left( \sum_{k=1}^m \Vert A_k\Vert\right)^3\exp\left(\frac{1}{n}\sum_{k=1}^m {\Vert A_k\Vert}\right).	
    \end{align} 
\end{thm}
\begin{proof}
	Let
\[ C=G\left(\frac{i}{n}\right)\quad \mbox{and}\quad  
		D=J_2\left(\frac{i}{n}\right).\]
Then
\[\Vert C\Vert\leq 1 \quad \mbox{and}\quad 
		\Vert D\Vert\leq 1. \]

Using method for 
 $E$ and $F$ in the proof of \Cref{thm:Trotter:first:JordanBanach} above, 
	\begin{align}
	\left\Vert C^n-D^n \right\Vert 
        \leq n \Vert C-D\Vert 
		\leq n \left(\left\Vert C-T_2\left(\frac{i}{n}\right)\right\Vert +\left\Vert D-T_2\left(\frac{i}{n}\right)\right\Vert\right).\label{ieqt:s:upper:s}
	\end{align}
 
Using an argument similar to the one in the proof of \Cref{thm:Trotter:first:JordanBanach}, we have 
	\begin{align}
		\left\Vert C-T_2\left(\frac{i}{n}\right)\right\Vert
&\leq \frac{1}{3!\cdot n^3}\left(\sum_{k=1}^m \Vert A_k\Vert \right)^3\exp \left(\frac{1}{n}\sum_{k=1}^m \Vert A_k\Vert \right) \label{ieqt:upper:c-t2} \\
		\left\Vert D-T_2\left(\frac{i}{n}\right)\right\Vert 
		&\leq \frac{1}{3!\cdot n^3}\left(\sum_{k=1}^m \Vert A_k\Vert \right)^3\exp \left(\frac{1}{n}\sum_{k=1}^m \Vert A_k\Vert \right)  \label{ieqt:upper:d-t2}.
	\end{align}

Combining \cref{ieqt:s:upper:s}, \cref{ieqt:upper:c-t2} and \cref{ieqt:upper:d-t2}, we get the desired result. 
\end{proof}

The following result extends \cite[Theorem 3.3]{chehade:wang:wang:2024:suzuki} to the setting of JB*-algebras for unitaries, providing a sharp bound.  
   
\begin{thm}\label{thm:Trotter:second:JB*}
	For any finite number of self-adjoint elements $A_1, A_2,\cdots, A_m$ in a unital JB*-algebra $\mathcal{A},$
	\begin{align*}
	\left\Vert \exp \left(i \sum_{j=1}^m A_j\right)-\left[S_2\left(\frac{i}{n}\right)\right]^n\right \Vert	
		\leq
		\frac{3^{m-1}+1}{6n^2}\left( \sum_{k=1}^m \Vert A_k\Vert\right)^3\exp\left(\frac{1}{n}\sum_{k=1}^m \Vert A_k\Vert\right).
  \end{align*}
\end{thm}
\begin{proof}
	Let $\displaystyle 
		H=S_2\left(\frac{i}{n}\right).$ Then $\displaystyle \Vert H\Vert\leq 1.$
	 Following the approach in the proof of \Cref{thm:Trotter:first:JB*},
	\begin{align}
	\left\Vert G^n-H^n \right\Vert
        & \leq n(\Vert G-T_2\Vert + \Vert H-T_2\Vert) \label{ieqt:s:upper:t}.
	\end{align}
	
Using a similar argument as in the proof of \Cref{coro:Trotter:second:JordanBanach}, we obtain
\begin{align}
		\left\Vert C-T_2\left(\frac{i}{n}\right)\right\Vert
		&\leq \frac{1}{3!\cdot n^3}\left(\sum_{k=1}^m \Vert A_k\Vert \right)^3\exp \left(\frac{1}{n}\sum_{k=1}^m \Vert A_k\Vert \right) \label{ieqt:upper:g-t2} \\
		\left\Vert H-T_2\left(\frac{i}{n}\right)\right\Vert 
		&\leq 3^{m-1} \frac{1}{3!\cdot n^3}\left(\sum_{k=1}^m \Vert A_k\Vert \right)^3\exp \left(\frac{1}{n}\sum_{k=1}^m \Vert A_k\Vert \right)\label{ieqt:upper:h-t2}.
	\end{align}	
	Combining \cref{ieqt:s:upper:t}, \cref{ieqt:upper:g-t2} and \cref{ieqt:upper:h-t2}, we derive the desired result.
	\end{proof}

\section{Higher-Order Jordan-Trotter Product Formulas: Third- Order and Beyond}\label{sec:highOrder}
In this section, we consider higher order Jordan-Trotter product formulas, particularly focusing on third-order and higher formulations. These advanced formulas extend the conventional second-order Jordan-Trotter formulas in the previous sections, offering improved accuracy for simulating complex quantum systems. 

\subsection{Third-Order Jordan-Trotter Product Formulas}
Let $A, B$ be elements in a unital Jordan-Banach algebra $\Al.$ Then
$J_2(t),$ as defined in \cref{notion:J2}, is given by
\[
J_2(t)=\exp(tA)\circ \exp(tB),	
\]
and $S_2(t),$ as defined in \cref{notion:S2}, is given by 
\[
S_2(t)=	U_{\exp\left(\frac{t}{2}A \right)}(\exp\left(tB\right))
\]
We also define
\begin{align} 
\widetilde{S_2}(t)&:=U_{\exp\left(\frac{t}{2}B \right)}(\exp\left(tA\right)) \nonumber\\
Q_3(t)&:=\frac{2}{3}S_2(t)+\frac{2}{3}\widetilde{S_2}(t)-\frac{1}{3}J_2(t)\label{notion:Q3}\\
T^o_3(t)&:=	I + t(A + B) + t^2\frac{(A+B)^2}{2} + t^3\frac{(A+B)^3}{6} \nonumber \\
T^s_3(t)&:=I + t(A + B) + t^2\frac{(A+B)^2}{2}+t^3\left(\frac{A^3+B^3}{6}+\frac{A\circ B^2+(A\circ B)\circ A}{2}\right)\nonumber\\
\widetilde{T^s_3}(t)&:=I + t(A + B) + t^2\frac{(A+B)^2}{2}+t^3\left(\frac{A^3+B^3}{6}+\frac{B\circ A^2+(B\circ A)\circ B}{2}\right)\nonumber\\
T^j_3(t)&:=I + t(A + B) + t^2\frac{(A+B)^2}{2}+t^3\left(\frac{A^3+B^3}{6}+\frac{A\circ B^2+A^2\circ B}{2}\right)\nonumber
\end{align}

Using a linear combination of second-order Jordan-Trotter product formulas, we derive a third-order Jordan-Trotter product formula.
\begin{thm}\label{thm:3order:product:JordanBanach}
Let $A, B$ be two elements in a unital Jordan-Banach algebra $\Al$. Then 
\begin{align}\label{eqt:3order:lc}
 \exp\left(t(A+B)\right)=Q_3(t)+O(t^4).
\end{align}
Moreover, 
\begin{align}\label{ineqt:3order:error}
\Vert \exp\left(t(A+B)\right)-Q_3(t)\Vert \leq \frac{2}{9}\vert t\vert^4(\Vert A\Vert+\Vert B\Vert)^4\exp\big(\vert t\vert(\Vert A\Vert+\Vert B\Vert) \big)	
\end{align}
\end{thm}

\begin{proof}
The Taylor expansion of $ \displaystyle \exp\left(t(A + B)\right) $ up to order 3 is $\displaystyle T^o_3(t)$. By direct calculations,  the Taylor expansions of $\displaystyle S_2(t)$,  $\displaystyle \widetilde{S_2}(t)$, and $\displaystyle J_2(t)$ up to order 3 are $\displaystyle T^s_3(t)$, $\displaystyle \widetilde{T^s_3}(t)$, and $\displaystyle T^j_3(t)$, respectively.  

Observe that
\begin{align}\label{eqt:T3o:relation}
T^o_3(t)=\frac{2}{3}T^s_3(t)+\frac{2}{3}\widetilde{T^s_3}(t)-\frac{1}{3}T^j_3(t).	
\end{align}
This implies that the Taylor expansion of $\displaystyle Q_3(t)$ up to order 3 coincides with that of $\displaystyle \exp\big(t(A+B)\big),$ which is $T^o_3(t).$ Hence, \Cref{eqt:3order:lc} follows.

To prove \eqref{ineqt:3order:error}, note that 
\begin{align}\label{ineqt:3order:error1}
\Vert \exp\left(t(A+B)\right)-Q_3(t)\Vert\leq \Vert \exp\left(t(A+B)\right)-T^o_3(t)\Vert+	\Vert Q_3(t)-T_3^o(t) \Vert.
\end{align}

Using the Taylor series remainder, we estimate
\begin{align}
\Vert \exp\left(t(A+B)\right)-T^o_3(t)\Vert&=\left\Vert\sum_{k=4}^{\infty} \frac{[t(A+B)]^k}{k!}\right \Vert\leq \sum_{k=4}^{\infty}\frac{\vert t\vert^k(\Vert A\Vert +\Vert B\Vert )^k}{k!}\nonumber\\
&\leq \frac{\vert t\vert^4(\Vert A\Vert +\Vert B\Vert)^4}{4!}\exp\big(\vert t\vert (\Vert A\Vert +\Vert B\Vert)\big)\label{inqt:T3o}. 	
\end{align}

From \Cref{eqt:T3o:relation}, we deduce
\begin{align}
\Vert Q_3(t)-T_3^o(t) \Vert\leq \frac{2}{3}\Vert S_2(t)-T_3^s(t)\Vert+\frac{2}{3}\Vert \widetilde{S}_3(t)-\widetilde{T_3^s}(t)\Vert+\frac{1}{3}\Vert J_2(t)-T_3^j(t)\Vert	\label{inqt:Q3:T3o}.
\end{align}

Following an argument similar to the proof of the inequality (3.6) in \cite{chehade:wang:wang:2024:suzuki}, we get
\begin{align}
\Vert J_2(t)-T_3^j(t) \Vert \leq \sum_{j+l=4}^{\infty}\frac{\vert t\vert^{j+l} \Vert A\Vert^j\Vert B\Vert^l}{j!l!}	\leq \frac{\vert t\vert^4(\Vert A\Vert +\Vert B\Vert)^4}{4!}\exp\big(\vert t\vert (\Vert A\Vert +\Vert B\Vert)\big) \label{inqt:T3j} .
\end{align}

Let $\widetilde{P}_s(t)$ be the sum of all terms of degree $s$ of Taylor expansion of $S_2(t)$. Then
\begin{align}
\Vert S_2(t)-T_3^s(t)\Vert&=\left\Vert \sum_{s>3}\widetilde{P}_s(t) \right \Vert \leq 3\left[\exp\big(\vert t\vert(\Vert A\Vert +\Vert B\Vert)\big)-\left(\sum_{i=0}^3\frac{\vert t\vert^i(\Vert A\Vert +\Vert B\Vert)^i}{i!}\right)\right ]\nonumber\\
&\leq 3\frac{\vert t\vert^4(\Vert A\Vert +\Vert B\Vert)^4}{4!}\exp\big(\vert t\vert (\Vert A\Vert +\Vert B\Vert)\big)\label{inqt:T3s}.
\end{align}

Similarly, we obtain
\begin{align}
\Vert \widetilde{S_2}(t)-T_3^s(t)\Vert\leq  3\frac{\vert t\vert^4(\Vert A\Vert +\Vert B\Vert)^4}{4!}\exp\big(\vert t\vert (\Vert A\Vert +\Vert B\Vert)\big)\label{inqt:tT3s}.
\end{align}
Combining these inequalities \Cref{ineqt:3order:error1}, \Cref{inqt:T3o}, \Cref{inqt:Q3:T3o}, \Cref{inqt:T3j}, \Cref{inqt:T3s}, and  \Cref{inqt:tT3s}, we obtain \eqref{ineqt:3order:error}.
\end{proof}
    
Using the special Jordan product in Banach algebras, \Cref{eqt:3order:lc} in \Cref{thm:3order:product:JordanBanach} gives the following
\begin{cor}\label{cor:3order:product:associative}
Let $A$ and $B$ be any elements in a unital Banach algebra. 
Then 
\begin{align}
 \exp\left(t(A+B)\right)&=\frac{2}{3}\exp\left(t \frac{A}{2} \right)\exp\left(tB\right)\exp\left(t\frac{A}{2}\right)+\frac{2}{3}\exp\left(t\frac{B}{2}\right)\exp\left(tA\right)\exp\left(t\frac{B}{2}\right)\nonumber\\
&\quad\quad\quad\quad\quad\quad\quad\quad-\frac{1}{6}\exp(tA)\exp(tB)-\frac{1}{6} \exp(tB) \exp(tA) +O(t^4)\label{eqt:3order:product}.
 \end{align}
\end{cor}
Note that in the case that $A$ and $B$ are skew symmetric elements in a unital C*-algebra, $Q_3(t)$ on the right side of \Cref{eqt:3order:product} provides an example of a linear combination of unitaries similar to the one proposed in \cite{childs:wiebe:2012:hamiltonian}.

\subsection{Higher-Order Jordan-Trotter Product Formulas}

In \cite{suzuki:1990:fractal, suzuki:1991:general}, Suzuki developed higher-order Trotter product formulas using recursive construction for associative algebras. 
In this section, we establish higher-order Trotter product formulas, both symmetric and non-symmetric, in Jordan-Banach algebras, which are non-associative.

\begin{thm}\label{thm:Trotter:higher:JordanBanach}
    Let $A_1, A_2,\cdots, A_m$ be elements in a unital Jordan-Banach algebra $\mathcal{A}$. For $n\ge 3$, let $Q_{n-1}(t)$ be the $(n-1)$-th order approximation or decomposition of $\exp\left(t\sum_{k=1}^m A_k\right)$ , that is, 
    \begin{align} \label{eqt:n-1:approximate}
       \exp\left(t\sum_{k=1}^m A_k\right)=Q_{n-1}(t)+O(t^n),
    \end{align}
    where $Q_2(t)$ can be either $J_2(t)$ or $S_2(t)$ as defined in {\rm \cref{notion:J2}} and {\rm \cref{notion:S2}}, respectively; if $m$ is odd, $Q_2(t)$ may also be $Q_{S_2}(t)$ as defined in {\rm\cref{notion:QS2}}.
   
Suppose $Q_n(t)$ is constructed as follows:
\begin{align}\label{eq:non-symmetric:approximant}
    Q_n(t)= \circ \prod_{j=1}^r Q_{n-1}(c_{nj} t) :=\left[\left(Q_{n-1}(c_{n1} t)\circ Q_{n-1}(c_{n2} t)\right)\circ \cdots\right]\circ Q_{n-1}(c_{nr} t),   
\end{align}
where $r\geq 2$ and the scalars $\{c_{nj}\}$ are solutions to the following systems of equations: 
\begin{equation}\label{eqt:norder}
    \sum_{j=1}^r c_{nj}=1\quad\mbox{and}\quad\sum_{j=1}^r (c_{nj})^n=0.	
\end{equation}
Then  $Q_n(t)$ is an $n$-th order approximation, i.e., 
\begin{align} \label{eqt:norder:approximate}
       \exp\left(t\sum_{k=1}^m A_k\right)=Q_{n}(t)+O(t^{n+1}). 
    \end{align}
\end{thm}

\begin{proof}
Let $\displaystyle B=\sum_{k=1}^m A_k.$
Consider the Jordan-Banach subalgebra generated by $\displaystyle \left\{I, B\right\}$, which is associative. Assuming  $\displaystyle \sum_{j=1}^r c_{nj}=1,$ we have
\begin{align}\label{eqt:taylor:product}
 \exp\left(tB\right)=\circ \prod_{j=1}^r \exp\left(c_{nj} tB\right) 	
\end{align}
 
The sum of Taylor terms of 
$\displaystyle \circ \prod_{j=1}^r \exp\left(c_{nj} t B\right)$ of exact degree $n$ is given by:
\begin{align}\label{eqt:taylor:approximate n:exp}
 \sum_{\substack{p_1, \dots, p_r \neq n \\ p_1 + p_2 + \cdots + p_r = n}} \frac{c_{n1}^{p_1} \cdots c_{nr}^{p_r}}{p_1! \cdots p_r!}t^n B^n  + \sum_{j=1}^r \frac{c_{nj}^n}{n!}t^n B^n,	
\end{align}
where the Taylor terms of $\displaystyle \circ \prod_{j=1}^r \exp\left(c_{nj} t B \right)$ of exact degree $n$ consist of two parts: one is the product of Taylor terms of $\displaystyle \exp\left(c_{nj} t B\right)$ with degrees $p_j$ such that $p_1 + \cdots + p_r = n$ and $p_1, \dots, p_r \neq n$; the other is the Taylor term of each $\displaystyle \exp\left(c_{nj} t B\right)$ of exact degree $n$.

By assumptions, 
the $(n-1)$-th order approximation $Q_{n-1}(t)$ can be expressed as:
\begin{align}\label{eqt:taylor:qn-1}
Q_{n-1}(t) = \sum_{s=1}^{n-1} \frac{t^s}{s!} B^s + t^n P_n(A_1, \dots, A_m) + O(t^{n+1}),	
\end{align}
where $P_n(A_1, \dots, A_m)$ is the coefficient of $t^n$ in the Taylor polynomial for $Q_{n-1}(t)$ and, in particular, a homogeneous polynomial in
$A_1, ..., A_m$ of degree $n$. Then, according to \cref{eq:non-symmetric:approximant} and \cref{eqt:taylor:qn-1}, the sum of all Taylor terms of $Q_n(t)$ of degree $n$  is
\begin{equation}\label{eqt:taylor:approximate n}
 \sum_{\substack{p_1, \dots, p_r \neq n \\ p_1 + p_2 + \cdots + p_r = n}} \frac{c_{n1}^{p_1} \cdots c_{nr}^{p_r}}{p_1! \cdots p_r!} t^n B^n + \sum_{j=1}^r c_{nj}^n t^n P_n(A_1, \dots, A_m).	 
\end{equation}
Plugging  \Cref{eqt:taylor:approximate n:exp}, \Cref{eqt:taylor:qn-1}, and \Cref{eqt:taylor:approximate n} into \Cref{eqt:taylor:product}, we see that
whether 
$\displaystyle P_n(A_1, \dots, A_m) \neq \frac{B^n}{n!}$
or not,  
\begin{align*}
    \exp\left(t\sum_{k=1}^m A_k\right) = Q_n(t) + O(t^{n+1}) 
\end{align*}
will hold if $\displaystyle 
    \sum_{j=1}^r (c_{nj})^n = 0$.
For $r\geq 3$,  \Cref{eqt:norder} has real solutions, as demonstrated in  \cite[Eq. (16)]{suzuki:1990:fractal}.
\end{proof}

The following result is on symmetric higher-order Jordan-Trotter product formula, which is different from the non-symmetric one in \Cref{thm:Trotter:higher:JordanBanach}.

\begin{thm}\label{thm:Trotter:higher:symmetric:JordanBanach}
Let $A_1, A_2,\cdots, A_m$ be elements of a unital Jordan-Banach algebra $\mathcal{A}.$ Suppose $\widetilde{Q}_{2}(t)=S_2(t),$ and for $n\ge 3$, let $\widetilde{Q}_{n-1}(t)$ be the $(n-1)$-th order approximation of $\exp\left(t\sum_{k=1}^m A_k\right)$, i.e., 
    \begin{align}\label{eqt:n-1order:tqn-1}
       \exp\left(t\sum_{k=1}^m A_k\right)=\widetilde{Q}_{n-1}(t)+O(t^n). 
    \end{align}
Define $\widetilde{Q}_n(t)$ recursively as follows:
\begin{align}\label{eqt:norder:appro:constr}
 \widetilde{Q}_n(t):=\left\{\widetilde{Q}_{n-1}\left(d_{nl} t\right)\cdots\left\{\widetilde{Q}_{n-1}\left(d_{n2} t\right)\widetilde{Q}_{n-1}(d_{n1} t)\widetilde{Q}_{n-1}\left(d_{n2} t\right)\right\}\cdots \widetilde{Q}_{n-1}\left(d_{nl}t\right)\right\},   
\end{align}
where the scalars $\{d_{nj}\}$ satisfy the following conditions:
\begin{align}\label{eqt:norder:symmetric}
d_{n1}+2\left(\sum_{j=2}^l d_{nk}\right)=1\quad\mbox{and}\quad d_{n1}^n+2\sum_{j=2}^l d_{nj}^n=0.
\end{align}
Then $\widetilde{Q}_n(t)$ is an $n$-th order approximation, i.e., 
\begin{align}\label{eqt:norder:symmetric:appro}
     \exp\left(t\sum_{k=1}^m A_k\right)=\widetilde{Q}_{n}(t)+O(t^{n+1}). 
\end{align}
\end{thm}
\begin{proof}
The proof is similar to that of \Cref{thm:Trotter:higher:JordanBanach}. For $\displaystyle B=\sum_{k=1}^m A_k$, the sum of Taylor terms of \[ \left\{\exp\left(d_{nl} t B\right)\cdots\left\{\exp\left(d_{n2} tB\right)\exp\left(d_{n1} tB\right)\exp\left(d_{n2} tB\right)\right\}\cdots \exp\left(d_{nl} tB\right)\right\}\] of exact degree $n$ is given by:
\[
\sum_{\substack{p_1, p_2, \dots, p_{2l-1}\neq n \\ p_1 + \cdots + p_{2l-1} = n}} \frac{d_{n1}^{p_1} \cdot d_{n2}^{p_2+p_3}\cdots d_{nl}^{p_{2l-2}+p_{2l-1}}}{p_1! \cdot p_2!\cdots p_{2l-1}!}t^n B^n+\frac{d_{n1}^n}{n!}t^n B^n+\sum_{j=2}^l \frac{2d_{nj}^n}{n!}t^n B^n 
\]

As in \cref{eqt:taylor:qn-1}, 
the $(n-1)$-th order approximation $\widetilde{Q}_{n-1}(t)$ can be expressed as:
\begin{align}\label{eqt:taylor:smqn-1}
\widetilde{Q}_{n-1}(t) = \sum_{s=1}^{n-1} \frac{t^s}{s!} B^s + t^n \widetilde{P}_n(A_1, \dots, A_m) + O(t^{n+1}),	
\end{align}
where $\widetilde{P}_n(A_1, \dots, A_m)$ is the coefficient of $t^n$ in the Taylor polynomial for $\widetilde{Q}_{n-1}(t)$.

Then, according to \cref{eqt:norder:appro:constr} and \cref{eqt:taylor:smqn-1},   the sum of Taylor terms of $\displaystyle \widetilde{Q}_n(t)$ of exact degree $n$ is given by:
\[
 \sum_{\substack{p_1, p_2, \dots, p_{2l-1} \neq n \\ p_1 + \cdots + p_{2l-1} = n}} \frac{d_{n1}^{p_1} \cdot d_{n2}^{p_2+p_3}\cdots d_{nl}^{p_{2l-2}+p_{2l-1}}}{p_1! \cdot p_2!\cdots p_{2l-1}!} t^n B^n +\left(\frac{d_{n1}^n}{n!}+\sum_{j=2}^l \frac{2d_{nj}^n}{n!}\right)t^n \widetilde{P}_n(A_1,\dots, A_m).
\]

If $\displaystyle d_{n1}^n+2\sum_{j=2}^l d_{nj}^n=0,$ then \[
    \exp\left(t\sum_{k=1}^m A_k\right) = \widetilde{Q}_n(t) + O(t^{n+1})
\]
holds.
\end{proof}

\begin{rmk}\label{rmk:inverterbility:symmetric}

{\rm (1)} 
According to Theorem {\rm4.1.3} in {\rm \cite{cabrera:rodriguez:2014:non}}, $\widetilde{Q}_2(t)$, which is $S_2(t)$, is invertible with inverse $\displaystyle \widetilde{Q}_2(-t)$, i.e., $S_2(-t)$. Then, by the symmetric construction of $\widetilde{Q}_3(t)$ in {\rm \Cref{thm:Trotter:higher:symmetric:JordanBanach}} above and Theorem {\rm4.1.3} in {\rm \cite{cabrera:rodriguez:2014:non}}, $\widetilde{Q}_3(t)$ is invertible with inverse $\displaystyle \widetilde{Q}_3(-t)$. 

Similarly, $\widetilde{Q}_n(t)$ is invertible with inverse $\displaystyle \widetilde{Q}_n(-t)$. 

{\rm (2)} If the invertibility of $\widetilde{Q}_n(t)$ is not required, then the non-symmetric form of $J_2(t)$ can also be used for $\widetilde{Q}_2(t)$ in {\rm \Cref{thm:Trotter:higher:symmetric:JordanBanach}}. Moreover, when $m$ is odd, $Q_{S_2}(t)$ is also a valid choice for $\widetilde{Q}_2(t)$ in {\rm \Cref{thm:Trotter:higher:symmetric:JordanBanach}}.

\end{rmk}

\begin{thm}\label{thm:Trotter:higher:symmetry:JordanBanach}
The symmetric $(2l-1)$-th order approximation $\widetilde{Q}_{2l-1}(t)$, i.e., 
 \begin{align*}
       \exp\left(t\sum_{k=1}^m A_k\right) = \widetilde{Q}_{2l-1}(t) + O(t^{2l}),
    \end{align*}
is accurate up to order $t^{2l}$, that is, 
    \begin{align*}
   \widetilde{Q}_{2l-1}(t) = \widetilde{Q}_{2l}(t).
    \end{align*}
\end{thm}

\begin{proof}
Assume
\[
\exp\left(t\sum_{k=1}^m A_k\right) = \widetilde{Q}_{2l-1}(t) + t^{2l} R_{2l}(A_1, \dots, A_m) + O(t^{2l+1}).
\]

Then, since
\[
I= \exp\left(t\sum_{k=1}^m A_k\right) \circ \exp\left(-t\sum_{k=1}^m A_k\right),\quad I=\widetilde{Q}_{2l-1}(t)\circ \widetilde{Q}_{2l-1}(-t) 
\]
we have
\[I=I+t^{2l} R_{2l}(A_1,\dots, A_m)\circ \widetilde{Q}_{2l-1}(-t)+(-t)^{2l}\widetilde{Q}_{2l-1}(t)\circ R_{2l}(A_1,\dots, A_m)+O(t^{2l+1}). \]

Hence,
\begin{align*}
R_{2l}(A_1, \dots, A_m) \circ \widetilde{Q}_{2l-1}(-t) + R_{2l}(A_1, \dots, A_m) \circ \widetilde{Q}_{2l-1}(t) = O(t).
\end{align*}

By definition in \Cref{eqt:norder:appro:constr}, we have $\widetilde{Q}_{2l-1}(0) = I$, $\displaystyle R_{2l}(A_1, \dots, A_m)=0.$
Thus, the approximation $\widetilde{Q}_{2l-1}(t)$ is accurate up to order $t^{2l}$.
\end{proof}

\begin{rmk}\label{rmk:product:new} Suppose the Jordan products are special and $A_1, A_2, \dots, A_m$ are operators. Then {\rm \Cref{thm:Trotter:higher:JordanBanach}} and {\rm \Cref{thm:Trotter:higher:symmetric:JordanBanach}} provide two higher-order Jordan-Trotter product formulas for associtive algebras. These formulas are analogous to the second-order Jordan-Trotter product formulas given in {\rm \cite[Theorem 2.1 and Theorem 2.3]{wang:2023:suzuki}} and offer additional approaches to simulate the quantum system. \end{rmk}

\section{Error Analysis and Visualization of Product Formula Approximations}\label{sec:example}

In this section, we present a comprehensive evaluation of the Jordan-Trotter product formula introduced in Eq. \Cref{eqt:3order:product} by analyzing its error behavior from two different perspectives. This analysis highlights the formula's  advantages over traditional Trotter-Suzuki decompositions, including the third-order symmetric formula.

First, in \Cref{sec:analyticalerrors}, we quantify the error between the exact unitary and the approximate unitary generated by the product formulas. This offers a precise mathematical measure of the accuracy of the Jordan-Trotter formula compared to conventional Trotter-Suzuki methods.
Second, in~\Cref{sec:statefidelity}, we evaluate the practical impact of these errors by applying the unitaries to a given quantum state and computing the fidelity between the resulting state and the exact evolved state. This approach offers a physical perspective on the effectiveness of product formulas, illustrating their utility in realistic quantum simulation scenarios. 
The Jordan-Trotter approximation is implemented in~\Cref{sec:example} using a \textit{linear combination of unitaries} (LCU) quantum circuit \cite{childs:wiebe:2012:hamiltonian}, which efficiently encodes the structure of the product formula. Additionally, contour and line plots illustrate the parameter regimes where the Jordan-Trotter formula outperforms traditional approaches, offering valuable visual insights into its performance.
Together, these analyses provide both theoretical and practical insights into the performance of the Jordan-Trotter product formula.

\subsection{Error Analysis of Product Formulas for Hamiltonian Evolution}\label{sec:analyticalerrors}

The error associated with the higher-order Trotter-Suzuki decomposition of a Hamiltonian $H = A+B$ depends on the order of the decomposition.
For a general Hamiltonian $H = A+B$ on a finite dimensional space, the Trotter-Suzuki formulas up to the fourth order are expressed as follows:
\begin{align}
 J_1(t)&=\exp(-itA)\exp(-itB) \label{eq:trotter:1}\\
  S_2(t)&=\exp\left(-it\frac{A}{2}\right)\exp(-itB)\exp\left(-it\frac{A}{2}\right)\label{eq:trotter:2}\\
S_3(t)&=S_2\left((2-2^{1/3})^{-1}t\right)S_2\left(-2^{1/3}(2-2^{1/3})^{-1}t\right)S_2\left((2-2^{1/3})^{-1}t\right)\label{eq:trotter:3}
\end{align} 
Here,  $J_1(t)$ is the first-order approximation,  $S_2(t)$ is the second-order approximation, and $ S_3(t)$ provides both third- and fourth-order accuracy due to its symmetric construction.

To quantify the error for different methods, we measure the difference between the true and approximate time evolution using Frobenius norm:
\begin{equation}\label{eq:error:unitaries}
    {\rm Error} := \left\Vert \widetilde{U}(t) - \exp\left({-itH}\right)\right\Vert_{\rm F}
\end{equation}
where $\widetilde{U}(t)$ is either $J_1(t)$ in \Cref{eq:trotter:1}, $S_2(t)$ in \Cref{eq:trotter:2}, $S_3(t)$ in \Cref{eq:trotter:3}, or $Q_3(t)$ in \cref{notion:Q3}, and $\exp\left({-itH}\right)$ is the true evolution operator. 

To evaluate and compare the performance of these product formulas, we consider the Hamiltonian $H = d_{1}X + d_{2}Y$, where $X$ and $Y$ are Pauli matrices. Contour plots (Figure~\ref{fig:contour_erorr_plots}) were generated to depict the Frobenius norm of the error as a function of $td_1$ and $td_2$. These plots provide a detailed comparison between the third-order Trotter Suzuki formula and the Jordan-Trotter product formula, highlighting their respective advantages across different parameter regimes.

\begin{figure}[!ht]
    \centering
    \includegraphics[width=0.99\textwidth]{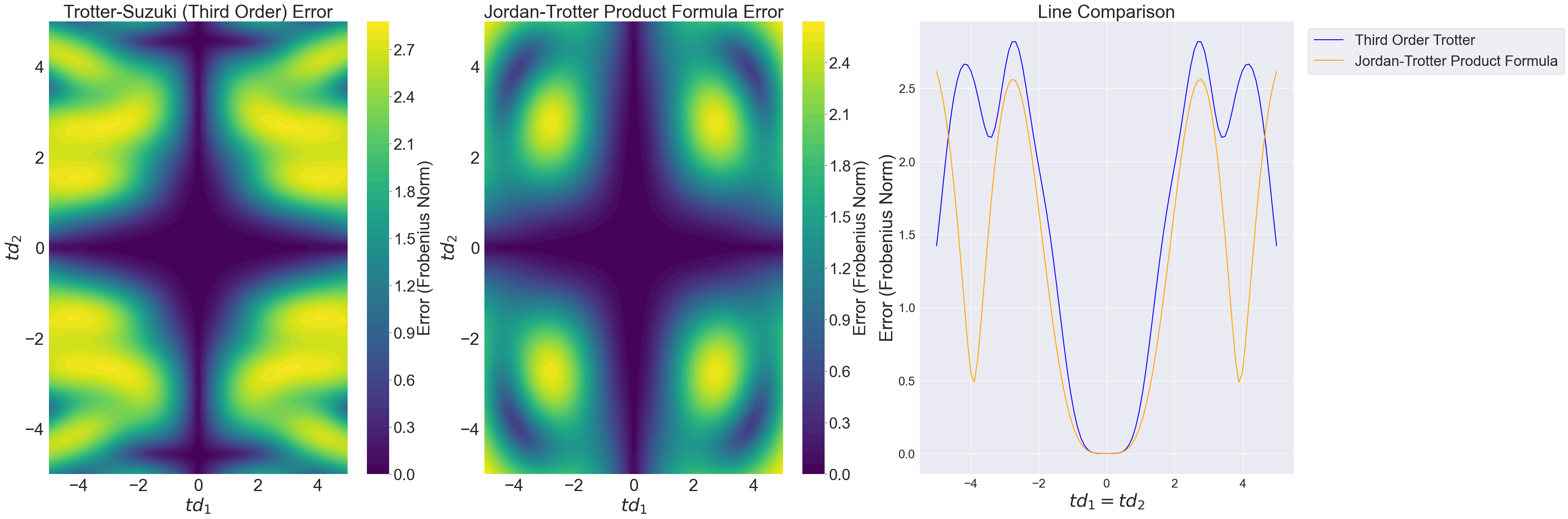}
    \caption{
    Comparison of the error associated with the third-order Trotter-Suzuki decomposition and the Jordan-Trotter product formula for the Hamiltonian $H= d_1 X + d_2 Y$. The left and middle panels show contour plots of the Frobenius norm of the error as a function of $td_1$ and $td_2$ for the third-order Trotter-Suzuki decomposition and the Jordan-Trotter product formula, respectively. The rightmost panel provides a line plot comparing the errors of both methods along the diagonal $td_1 = td_2$.
    }
    \label{fig:contour_erorr_plots}
\end{figure}

The results of the plots reveal key insights into the error behavior of the third-order Trotter-Suzuki decomposition and the Jordan-Trotter product formula. Both methods exhibit a small error along the $td_1 = 0$ and $td_2=0$ lines, which is expected as the effective Hamiltonian approaches a single operator along these axes. However, the Jordan-Trotter product formula demonstrates significantly smaller errors across a wider region of the parameter space compared to the third-order Trotter-Suzuki decomposition. This is particularly evident in the line plot (Figure~\ref{fig:contour_erorr_plots} - right), where the Jordan-Trotter product formula shows markedly reduced errors in the vicinity of $td_1 = td_2 = \pm 4$. Additionally, the periodicity observed in the error estimates for both methods aligns with the sinusoidal forms that the Pauli matrices acquire when exponentiated. This periodic behavior underscores the mathematical structure of the approximations and highlights the regions where the Jordan-Trotter formula achieves superior accuracy. These findings emphasize the potential of the Jordan-Trotter product formula for applications requiring reduced error in broader parameter regimes.


\subsection{State Fidelity Analysis under Product Formula Approximations}\label{sec:statefidelity}

In this subsection, we evaluate the errors introduced by the product formula approximations acting on a quantum state $\ket{\psi(0)}$. 
Specifically for Pauli matrices Z and X, let
\begin{align}
    H = \alpha Z + \beta X,
\end{align}
\noindent
and for $A=aZ$ and $B=\beta X$, we consider
\begin{align}\label{eqt:error:state}
   \varepsilon:= \left\Vert\widetilde{U}(t)\ket{\psi(0)}-\exp(-iHt)\ket{\psi(0)} \right\Vert 
\end{align}
where the approximation $\widetilde{U}(t)$ is either $J_1(t)$ in \Cref{eq:trotter:1}, $S_2(t)$ in \Cref{eq:trotter:2}, $S_3(t)$ in \Cref{eq:trotter:3}, or $Q_3(t)$ in \cref{notion:Q3}, and $\exp\left({-itH}\right)$ is the exact time-evolution operator. To illustrate the differences and strengths of the methods, we compare the approach introduced in this paper with the well-established Trotterization technique. We consider the simple case of a single-qubit Hamiltonian, which provides a clear and accessible framework for understanding the cost and error scaling of each method.  

One can show 
\begin{align}
   \exp({-itH}) = cos\left(t\sqrt{\alpha^2 + \beta^2}\right)\cdot I - \frac{i}{\sqrt{\alpha^2 + \beta^2}} sin\left(t\sqrt{\alpha^2 + \beta^2}\right)\cdot H
\end{align}

For simplicity, and without loss of generality, we choose the initial state $\ket{\psi(0)}=\ket{0}$, which corresponds to the quantum state where all qubits are initialized to 0. 
\begin{figure}[ht]
    \centering
    \includegraphics[width=0.95\textwidth]{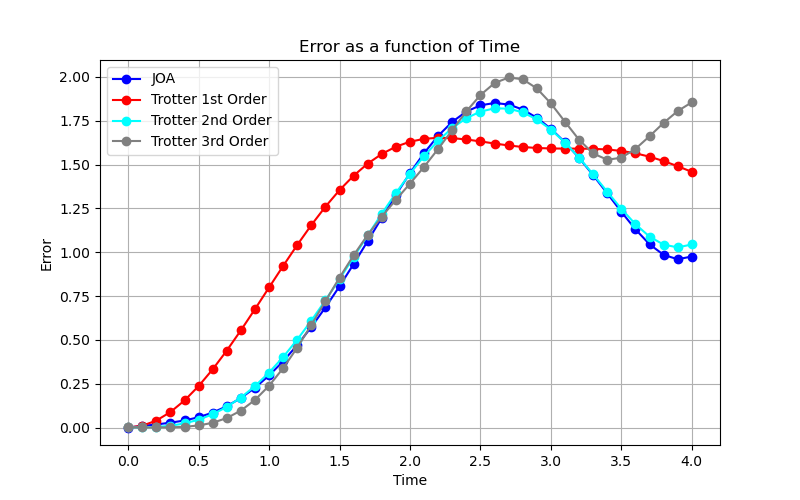}
    \caption{Comparison of the time evolution accuracy for the first-order(red), second-order(cyan), and third-order(gray) Trotter-Suzuki decompositions, along with the Jordan-Trotter product formula(darker blue), against the exact result for $H =  Z +  X$. The plot shows the deviation from the exact solution as a function of evolution time $t$, highlighting the improved accuracy of higher-order approximations.} 
    \label{fig:comparison}
\end{figure}

In Fig.~\ref{fig:comparison}, we illustrate the error scaling behavior of the Trotter-Suzuki decomposition (at first, second,  and  third order) as applied to the simulation of a single-qubit Hamiltonian $H = \alpha Z + \beta X$ for $\alpha=\beta=1$. The $x$-axis represents the simulation the $t$, which is the parameter governing the evolution of a quantum state under the Hamiltonian $H$. In physical terms, $t$ corresponds to how long the system evolves under the influence of $H$, with larger $t$ leading to a more pronounced evolution. In the context of quantum simulation, $t$ also determines the challenge of accurately approximating the time evolution.

\newpage 

\section{Acknowledgments}

The authors would like to thank the referee for the very careful reading of our manuscript, and for his or her numerous comments and suggestions, which greatly enhanced the quality and readability of the paper.

S. C. is supported by the DOE Advanced Scientific Computing Research (ASCR) Accelerated Research in Quantum Computing (ARQC) Program under field work proposal 3ERKJ354.  A.D. is supported by DOE Office of Nuclear Physics (NP) through the ``Quantum learning for accelerated nuclear science'' (QLASS) project under FWP ERKBP91.  This work was initiated during Z.W.'s visit to Oak Ridge National Laboratory in Summer 2024, supported by the DOE Office of High Energy Physics (HEP) under grant ERKAP89. Z.W. extends his gratitude to Oak Ridge National Laboratory for their warm hospitality and generous support during the visit.

\bigskip
\noindent 
{\bf Declaration of competing interest:}\
The authors declare that they have no known competing financial interests or personal relationships that could have appeared to influence the work reported in this paper.

\bigskip
\noindent 
{\bf Data availability:} \
The data supporting this study are available upon reasonable request.

\end{document}